\RequirePackage[orthodox,l2tabu]{nag}
\documentclass[a4paper, 8pt]{extarticle}

\usepackage[utf8x]{inputenc}
\usepackage[T1]{fontenc}
\usepackage{amsmath,amsfonts}
\usepackage{url}
\usepackage[english]{babel}

\usepackage{graphicx}
\usepackage{algorithm}
\usepackage[noend]{algpseudocode}

\usepackage{xspace}
\usepackage{url}

\usepackage{varioref}

\usepackage{tikz}
\usepackage{xcolor}
\usepackage{csquotes}

\usepackage{tabularx}
\usepackage{multirow}
\newcolumntype{Y}{>{\centering\arraybackslash}X}
\usepackage{makecell}

\usepackage{booktabs} % For formal tables

\usepackage{amsthm}

\newcommand*{\FN}{\operatorname{FN}}
\newcommand*{\FP}{\operatorname{FP}}

\newcommand*{\FNR}{\operatorname{FNR}}
\newcommand*{\FPR}{\operatorname{FPR}}

\newtheorem{theorem}{Theorem}[section]
\newtheorem{claim}{Claim}[section]

\theoremstyle{definition}
\newtheorem{definition}{Definition}[section]

 \newcommand*{\DUPLICATE}{\textsf{DUPLICATE}\xspace}
 \newcommand*{\UNSEEN}{\textsf{UNSEEN}\xspace}

\usepackage{authblk}
\usepackage{fullpage}
\usepackage{multicol}
\usepackage{float}

\begin{document}
\title{Quotient Hash Tables - Efficiently Detecting Duplicates in Streaming Data}

\author[a,c]{R\'emi G\'eraud}
\author[a,b]{Marius Lombard-Platet\thanks{Part of this work was done while the author was visiting Ingenico Labs}}
	\author[a,c]{David Naccache}
	\affil[a]{D\'epartement d'informatique de l'ENS, \'Ecole normale sup\'erieure, CNRS, PSL Research University, Paris, France}
	\affil[b]{Be-Studys, Geneva, Switzerland}
	\affil[c]{Ingenico Labs Advanced Research, Paris, France}
\setcounter{Maxaffil}{0}
\renewcommand\Affilfont{\itshape\small}

\date{}
\maketitle

\begin{abstract}

This article presents the Quotient Hash Table (QHT) a new data structure for duplicate detection in unbounded streams. QHTs stem from a corrected analysis of streaming quotient filters (SQFs), resulting in a 33\% reduction in memory usage for equal performance. We provide a new and thorough analysis of both algorithms, with results of interest to other existing constructions.

We also introduce an optimised version of our new data structure dubbed Queued QHT with Duplicates (QQHTD).

Finally we discuss the effect of adversarial inputs for hash-based duplicate filters similar to QHT. 
\end{abstract}

\begin{multicols}{2}

	\section{Introduction and Motivation}
	We consider in this paper the following problem: given a possibly infinite stream of symbols, detect whether a given symbol appeared somewhere in the stream. It turns out that instances of this \emph{duplicate detection} problem arise naturally in many applications: backup systems
	~\cite{Fu15}
	or Web caches~\cite{Fan00}, search engine databases or click counting in web advertisement
	~\cite{Metwally:2005:DDC:1060745.1060753}
	, retrieval algorithms~\cite{Borg:2014:RSD:2652524.2652556}, data stream management systems~\cite{Babcock:2004:LSA:977401.978165}, or even cryptographic contexts where nonce-reuse is problematic\footnote{A nonce is a one-time use random number.}
	\cite{nonce}.

	It is generally not possible to store the whole stream in memory, therefore practical solutions to this problem must somehow trade off memory for accuracy. The special case of a single duplicate
	has a known optimal solution~\cite{one_dup,one_dup_opt}, but to the best of the authors knowledge no such results exists for detecting \emph{all} duplicates in an unbounded stream.

    Several algorithms were proposed to address this specific question, which we discuss below. Of particular interest to our investigation are Streaming Quotient Filter (SQF), as described by Dutta et al. in~\cite{Dut13}. We point out several crucial mistakes in the original analysis of SQF, and in doing so highlight that a more efficient data structure can be constructed along the same lines.
    We flesh out such a data structure, which we call the Quotient Hash Table (QHT), and provide
    a thorough analysis of both SQF and QHT. This analysis is our main contribution. QHT itself is not optimal and can be further improved: we describe such an improvement, dubbed Queued QHT with Duplicates (QQHT), and benchmark our new algorithms against popular alternatives.
	
	\section{Preliminaries and related work}
	\subsection{Duplicate Detection}\label{sub:dup_detection}
	Let $E = (e_1, \dotsc, e_n, \dotsc)$ be a (possibly infinite) sequence of elements $e_i \in \Gamma$, and write $U = |\Gamma|$. An element $e_i$ from $E$ is a \emph{duplicate} if $\exists e_j \in E, j < i$ such as $e_j = e_i$. Otherwise, $e_i$ is \emph{unseen}.\footnote{Note that by definition $e_1$ is always unseen.}
    
	The \emph{duplicate detection problem} is the question of finding all duplicates in a stream $E$.\footnote{Note that this problem is equivalent to a dynamic formulation of the approximate set membership problem.} We recall the following well-known result:
	\begin{theorem}
	Assume that each $e_i$ is sampled uniformly at random from $\Gamma$.
	Then perfect detection requires $U$ memory bits.
	\end{theorem}
	\begin{proof}
	A perfect duplicate filter must be able to store all streams $E_i = (e_1, \dotsc, e_n)$, i.e. must be able to store any subset of $\Gamma$, which we denote by $\mathcal P(\Gamma)$. Given that there are $|\mathcal P(\Gamma)| = 2^{|\Gamma|} = 2^U$ of them, according to information theory any such filter requires at least $\log_2(|\mathcal P(\Gamma)|) = \log_2(2^U) = U$ bits of storage.

	\end{proof}
	Because of this result, perfect duplicate detection is often out of reach when $U$ is big --- however probabilistic solutions are often sufficient for many applications. Such algorithms make errors: false positives (claiming a duplicate where there isn't) and false negative (missing a duplicate). 
	
	\subsection{Filters}
    \begin{definition}[Filter]
        A filter $F$ over the memory $\mathcal M$ is a tuple $F = (\mathcal{S}, \textsf{Detect}, \textsf{Insert})$, where: 
        \begin{itemize}
        \item $\mathcal S \in \mathcal M$ is the \emph{current state}
        \item $\textsf{Detect}: \Gamma \times \mathcal M \to \{\DUPLICATE, \UNSEEN\}$ 
        \item $\textsf{Insert}: \Gamma \times \mathcal M \to \mathcal M$
        \end{itemize}
    \end{definition}
    Here \DUPLICATE corresponds to a guess that the provided element is duplicate, and \UNSEEN that it is unseen. \textsf{Insert} corresponds to an update of the filter's memory state after observing a new element. Here $\mathcal M$ models the amount of memory (states) available to the algorithm.
    
    In practice, \textsf{Detect} and \textsf{Insert} are often merged into a single algorithm $\textsf{Stream} \gets \textsf{Insert} \circ \textsf{Detect}$. 

\begin{definition}[False positive (resp. negative)]\label{def:false_positive}
	If, for an unseen (resp. duplicate) element $e$, $\textsf{Detect}(e)$ outputs \DUPLICATE (resp. \UNSEEN), $e$ is called a \emph{false positive} (resp. \emph{negative}).
\end{definition}
    
\begin{definition}[FPR, FNR]
    The \emph{false positive rate} (FPR) of a stream $E$ is the frequency of false positive. The \emph{false negative rate} is similarly defined as the frequency of false negatives.
\end{definition}

    We are interested in filters that whose FNR and FPR can be kept low when $\mathcal M$ is bounded.
    
    \subsection{Hash-based filters}
    Filters rely on hashing to efficiently answer the duplicate detection problem.
    These filters are often a variation over the well-known Bloom filters~\cite{Blo70}. A Bloom filter uses an array $T$ of $M$ bits, initially all set to $0$. $k$ hash functions $h_i$ are also needed. Insertion of an element $e$ is made by setting all $T[h_i(e)]$ to $1$. Detection of an element $f$ is the AND of all cells $T[h_i(f)]$: if the AND value is $0$, then emit \UNSEEN, otherwise emit \DUPLICATE.
	
	It is easy to see that this approach has a FNR of 0. However, as the stream grows, the FPR gets worse, and in the limit of an infinite stream the FPR is 1. Many variants have been proposed in the literature to compensate for this effect \cite{Tar12}, mostly by allowing deletion \cite{Coh03,Cor09}, but doing so increases the FNR.
	
	Alternatively, other filters prefer to store fingerprints: this is the rationale behind several recent constructions \cite{Dut13,Fan14}.
    
    For the data structures most interesting to our question, we refer the reader to the algorithms described in~\cite{Den06, Fan14, She08, Dut13, Yoo10}.\footnote{Cuckoo filters~\cite{Fan14} requires a minimal adaptation for unbounded streams: in the original paper, failure is emitted after some number of relocations; we just discard the failure. Theoretical analysis of this new structure is not addressed in this paper.}
   
   Other data structures, such as \cite{Che17, Guo10} are not considered in this paper, as they require an unbounded amount of memory as the stream grows.
	
	\subsection{Streaming Quotient Filter}
	One construction which we must describe at length is the Streaming Quotient Filter (SQFs)~\cite{Dut13}.  Given an element $e$, a certain fingerprint\footnote{\cite{Dut13} refers to them as \enquote{signatures}, but we shun this term to avoid any claim of cryptographic properties.} $s(e)$ is stored in an array, at row $h(e)$ and a certain column amongst $k$. This array constitutes the filter's state.
	
	The filter's construction uses integers $q$, $k$, $r$ $r' < r$ and a hash function $h: \{0, 1\}^* \to \{0, 1\}^{q+r}$. The filter's state is an array of $2^q$ rows and $k$ columns, each holding a $\sigma$-bit element (with $\sigma = r' + \lceil \log_2(r + 1) \rceil $), or the special empty symbol $\bot$. The filter's state is initially $\bot$ in every cell.
	
	Then, \cite{Dut13} describes $\mathsf{Stream}(e)$ as follows:
	\begin{itemize}
		\item Compute $h(e)$. The $q$ most significant bits of $h(e)$ define $h_e^q$, and the $r$ least significant bits define $h_e^r$.
        \item Let $w_e$ be the Hamming weight of $h_e^r$, and let $h_e^{r'}$ be $r'$ bits deterministically chosen\footnote{E.g., the $r'$ least significant bits of $h_e^r$.} from $h_e^r$. Let $s(e) \gets h_e^{r'} \| w_e$ where $\|$ denotes concatenation.
		
        \item If $s(e)$ is already stored in the row numbered $h_e^q$, emit \DUPLICATE.
		\item Otherwise, store it in one empty cell of that row. If no empty cell exists in the row, store $s(e)$ in one \emph{random} cell of the row, replacing any fingerprint previously stored there. Emit \UNSEEN.
	\end{itemize}
	
	\section{Revisiting SQF: Quotient Hash Table}\label{sub:sqf_errors}
SQF is introduced and analysed in \cite{Dut13}. However, a careful reading reveals several crucial mistakes in that analysis. We focus here on two of them that directly impact the claim of SQF near-optimality. 
Note that here, as in the rest of this paper, hash functions are modelled as pseudo-random functions.
\begin{enumerate}
    \item Cells in the filter's state are not independent: in particular, in every row, non-empty cells hold different values (by design); 
    \item Terms in geometric sums of order $> 1$ cannot be neglected.
\end{enumerate}
Taking these effects into account (see Sections~\ref{sub:FPR} and \ref{sub:FNR}), and using the same approximation than \cite{Dut13} used\footnote{Even though not mentioned in \cite{Dut13}, the approximation (stemming from Catalan numbers) is only asmptotatically valid, but computations show that the approximation is correct even for small values of $r$, such as the ones used in practical applications.}, $\binom{2r}{r} \approx \frac{4^r}{\sqrt{\pi r}}$, we get the following asymptotic formulae for the FNR and FPR
\begin{align*}
    \FPR_\infty \approx \frac{k}{2^{r'}\sqrt{\pi r}} \qquad \FNR_{\infty} \approx 1 - \frac{k}{2^{r'}\sqrt{\pi r}}
\end{align*}
that disagree with \cite{Dut13} --- most importantly, the error rates do not decrease to $0$, contrarily to what was claimed. Interestingly, the suggested parameters ($r= 2$, $k = 4$, $r' = r/2 = 1$) indeed achieve an FNR of $0$ --- but also an FPR of 1.\footnote{With these values, only $4$ distinct fingerprints exist, and they can all be stored in the $4$ cells per row. When the filter is full, any duplicate will be reported as such, hence a FNR of $0$. But every new element will also be reported as duplicate, hence a FPR of $1$.}
    
    Further more, there is redundancy between the Hamming weight $w_e$ of $h_e^r$ and the reduced remainder $h_e^{r'}$. For instance, if $h_e^{r'}$ contains at least one bit set to $1$, we know that $w_e \neq 0$. Intuitively this means SQF is wasteful, and we could expect to avoid collisions in the filter's state by using a better adjusted encoding: this intuition happens to be correct as shown in Section~\ref{sub:application_sqf}.
    
    Finally, since the state table contains $2^q$ rows, with $k$ cell each, the total memory required by an SQF is $M = 2^q k \sigma$.
	
	\subsection{Full-size Hashing, Memory Adjustments}
	Our first observation is that SQF's fingerprint scheme (a hash and a Hamming weight) can be fruitfully replaced by a single hash function of the same size. Not only does this simplify the theoretical analysis, it also provides a much more efficient use of the available space. We also use more flexible hash functions, that give much more flexibility in adjusting the total memory $M$ of the filter. Combining these two effects, we obtain Algorithm~\ref{a:QHT} which we call Quotient Hash Table (QHT).
	
		\begin{algorithm}[H]
			\caption{QHT Setup and Stream}\label{a:QHT}
			
			\begin{algorithmic}[1]
			\Function{Setup}{$M, \sigma, k$}
			    \Comment $M > \sigma k$ and $0 < k \leq 2^\sigma$ 
				\State Let $N \gets \left\lfloor {M} / (k \cdot \sigma) \right\rfloor$.
				\State Choose a hash function $h$ over $\left[0, N - 1\right]$
				\State Choose a hash function $s$ over $[0, 2^\sigma - 1]$
				\State Let $T$ be a $N \times k$ array with $\sigma$-bit cells, initialized to $\bot$.
			\EndFunction 
			\end{algorithmic}
			
            \begin{algorithmic}[1]
                \Function{Stream}{$e$}
				\For{each cell $b_i$ in row $T[h(e)]$}
				\If{ $b_i = s(e)$}
				\State \Return $\DUPLICATE$
				\EndIf
				\EndFor
				\State Let $b_\emptyset$ be the first empty cell in row $T[h(e)]$.
				\If{$b_\emptyset$ does not exist}
				\State $b_\emptyset \stackrel{\$}{\gets} T[h(e)]$
				\EndIf
				\State Store $s(e)$ in bucket $b_\emptyset$
				%\EndIf
				\State \Return $\UNSEEN$
 				\EndFunction
			\end{algorithmic}
		\end{algorithm}

    \subsection{Empty Cells}\label{sub:empty}
    SQF and QHT as described above make essential use of the \enquote{empty} cells in $T$. The need for this feature is present for all fingerprint-based structures including Cuckoo filters \cite{Fan14}.\footnote{Other constructions do not face this issue, including SBF \cite{Den06} and b\_DBF \cite{She08}: in these schemes, $0$ always codes for absence.} However, a low-level implementation cannot rely on the availability of such a special value. Our options are to initialise all cells to $0$ and either
    
    \begin{enumerate}
    \item\label{item:empty_1} treat $0$ as a fingerprint;
    \item\label{item:empty_2} if an element has a fingerprint of $0$, reassign its fingerprint to $1$;
    \item\label{item:empty_3} while an element has a fingerprint $0$, compute a new fingerprint based on some deterministic scheme.
    \end{enumerate}
    The first option is at a risk of a high false positive rate, even for small streams: when a new element, whose fingerprint is $0$, should be stored in an empty bucket, it is instead dismissed as a duplicate. More specifically, before the filter is completely filled, a new element has probability at least $\frac{1}{S}$ to be false positive, which leads to a high number of false positive at the beginning of any stream. 
    
    The second option is significantly faster than the third\footnote{The third option needs, on average, $\frac{S}{S-1}$ hash computations for each insertion, whereas the second option only needs 1.}, which on the other hand has better statistical properties, making the analysis simpler. While the second option was preferred by \cite{Fan14} in their implementation\footnote{\url{https://github.com/efficient/cuckoofilter/blob/aac6569cf30f0dfcf39edec1799fc3f8d6f594da/src/cuckoofilter.h}}, we retain the slower, but easier to analyse option.

    \subsection{Semi-sorting}
    The technique of \emph{semi-sorting} was introduced in \cite{Fan14} to shave some extra storage.
    The idea is as follows: treating empty cells as buckets containing the \enquote{0} fingerprint, for each row, sort the cells by their fingerprints, and then encode the result as a number. 
    
    As an example, for $s=4$ and $k=4$ there are only 3,876 possible sorted states, which can be encoded using $12$ bits, as opposed to the $16$ bits required to store four 4-bit fingerprints.
    
    \subsection{Comparison with Hash Tables}
	As the name implies, QHT are related to hash tables. Indeed, a hash table is a QHT, wherein the number of rows is equal to $1$. In particular, QHT cannot have worse performances than such structures, and bear similarities with so-called \emph{compaction} techniques \cite{10.1007/3-540-56922-7_6}.

    \section{Error Rate Analysis}

	\subsection{False Positive Rate}\label{sub:FPR}
	
	Consider a QHT with $N$ rows, $k$ buckets per row, and $s$-bits fingerprints. For simplicity, further assume that no bucket is empty (which is true after some time), and that the stream is sampled uniformly at random from $\Gamma$.
	
	\begin{theorem}
	For a QHT of $N$ rows, $k$ buckets per row and $S$ possible fingerprints, as the number of insertions goes to infinity, the FPR goes to $k/S$. Moreover the probability that an unseen element inserted after $m$ other elements triggers a false positive is $\displaystyle \FP_m = \frac{k}{S}\left(1 - \left[1 - \frac{1}{Nk}\right]^m\right)$.
	\end{theorem}

\begin{proof}
	An element $e$ is a false positive if and only if $e$ has not been encountered yet, but $\textsf{Detect}(e) = \DUPLICATE$. This event is triggered by the presence, in the filter, of another element $e'$ with a hash and fingerprint colliding with those of $e$. $e'$ is called a \emph{false duplicate}, and if $e'$ is still in the filter when $e$ arrives, we refer to the event as a \emph{hard collision}.
	
	Our first remark is that the only false duplicate that may create a hard collision with $e$ is the last false duplicate inserted before $e$ arrives: let us assume that $e_1$, $e_2$ are false duplicates, and that $e_1$ arrives before $e_2$. When we insert $e_2$ in the filter, $e_1$ is either still in the filter or has been evicted.
    \begin{itemize}
    \item If $e_1$ has been evicted, then $e_1$ will not hardly collide with $e$.
    \item If $e_1$ is still in the filter, then $e_2$ will be claimed as a duplicate and dismissed.
    \end{itemize}
    However, if we look at the table $T$ storing all fingerprints, dismissing $e_2$ is strictly equivalent to replacing $e_1$ by $e_2$. Consequently, every false duplicate is erased by any new false duplicate, and only the last false duplicate can cause a hard collision. As a result, we will only focus on the probability that the last false duplicate (before $e$ arrives) causes a hard collision.
	
	Let us assume that the last false duplicate appears at position $i$ of the stream $E = \{ e_1, e_2, \dotsc, e_m, e \}$, in other words, $e_i$ is the last false duplicate in the stream before $e$.
	
	Now, $e_i$ has to remain in the filter until $e$ arrives, even though new elements are added. Let us suppose that element $e_j$, $i < j \leq m$, does not evict $e_i$ from the filter. For $e_i$ to \emph{be} evicted, the following conditions must be true: 
	\begin{itemize}
		\item $h(e_j) = h(e_i)$
		\item $s(e_j)$ is different from all the other fingerprints stored in the row $T[h(e_j)]$ (knowing that one of the buckets contains $s(e_i)$)
		\item $s(e_j)$ is inserted into the bucket in which $s(e_i)$ is stored
	\end{itemize}
    Since we know that $e_i$ is the last false duplicate, we cannot simultaneously have $h(e_j) = h(e_i)$ and $s(e_j) = s(e_i)$. As such, the first two conditions are not independent. Let $P_{\text{hs}}$ be the probability that these two conditions are satisfied. Given that $e_i$ is the last false positive, there are only $NS-1$ possibilities for the couple ($h(e_j), s(e_j))$. Moreover, among these $NS-1$ states, only $S-1$ verify the first condition, and among these $S-1$ states, only $S-k$ verify the second condition. Finally, $P_{\text{hs}} = \frac{S-k}{NS-1}$. If we assign the last event to the probability $P_{\text{selection}}$, we immediately get $P_{\text{selection}} = \frac{1}{k}$.
    
    Finally, the probability $P_{\neg\text{evict}}$ of $e_i$ not being evicted by $e_j$ is $P_{\neg \text{evict}} = 1 - P_{hs}P_{\text{selection}}$ and:
    $
    P_{\neg \text{evict}} = 1 - \frac{S-k}{k(NS-1)}
    $

	Now, $e_i$ has to avoid eviction by every element before $e$ arrives, i.e. by all elements $e_{i+1}, \dotsc, e_m$, which happens with probability $P({\text{hc}})_i = \left(P_{\neg\text{evict}}\right)^{m-i}$.
	
	At that point, we know the probability that a hard collision happens when the last false duplicate has been inserted at position $i$.
	
	The probability of any element $e'$ being a false duplicate is $P_\text{fd} = \frac{1}{N} \frac{1}{S}$. In the stream $E = (e_1, \dotsc, e_m, e)$, the last false duplicate is $e_m$ with probability $P_\text{fd}$; it is $e_{m-1}$ with probability $\left(1 - P_\text{fd}\right)P_\text{fd}$; and $e_{m-k}$ with probability $\left(1 - P_\text{fd}\right)^k P_\text{fd}$. The probability that the next element will result in a false positive after $m$ elements are inserted, is equal to the probability that the last false duplicate is not evicted. Thus the probability $\FP_m$ that a false positive happens after $m$ elements are inserted is 
	\begin{align*}
		\FP_m &= \sum_{j=1}^m (e_j\text{ is the last false duplicate } \times  
        \\& \qquad \qquad e_j \text{ is not evicted until } e \text{ arrives})\\
		&= \sum_{j=1}^m \left(1 - P_\text{fd}\right)^{m-j} P_\text{fd}\times P(\text{hc})_j\\
		&= \sum_{j=1}^m \left(1 - P_\text{fd}\right)^{m-j} P_\text{fd} \left(1 - \frac{S-k}{k(NS-1)}\right)^{m-j}\\
        &= P_\text{fd}\frac{1 - \left[\left(1 - P_\text{fd}\right)\left(1 - \frac{S-k}{k(NS-1)}\right)\right]^m}{1 - \left(1 - P_\text{fd}\right)\left(1 - \frac{S-k}{k(NS-1)}\right)}\\
	         \end{align*}
    \begin{align*}
        &=  \frac{1}{NS}\frac{1 - \left[\left(1 - \frac{1}{NS}\right)\left(1 - \frac{S-k}{k(NS-1)}\right)\right]^m}{1 - \left(1 - \frac{1}{NS}\right)\left(1 - \frac{S-k}{k(NS-1)}\right)}\\
        &= \frac{1-\left[1-\frac{S-k}{(NS-1)k} - \frac{1}{NS} + \frac{S-k}{kNS(NS-1)}\right]^m}{1+ \frac{NS(S-k)}{(NS-1)k}-\frac{1}{k}\frac{S-k}{NS-1}}\\
        &= \frac{1 - \left[1 - \frac{1}{NS} - \frac{S - k}{NSk}\right]^m}{1 + \frac{S - k}{k}} \\
        &= \frac{k}{S}\left(1 - \left[1 - \frac{1}{Nk}\right]^m\right)
	\end{align*}
	% C'est pas moi j'ai un alibi j'étais au cinéma <- olol
	We now have the probability that a new element $e$, inserted after $m$ insertions, is detected as a false positive.

    Assuming the stream is of size $n$, and noting $\FPR_n$ its false positive rate, we get that $\FPR_n$ is equal to the averages of all $\FP_m$ of its stream, i.e., $\displaystyle\FPR_n = \frac{1}{n} \sum_{m=1}^{n} \FP_m$. Given that $\displaystyle\lim_{m\to \infty} \FP_m = \frac{k}{S}$ and using Ces\`{a}ro's mean properties, we get, as one could have expected, 
	$\FPR_n \xrightarrow[n\to \infty]{}\frac{k}{S}$.
\end{proof}

Thanks to the expression of $\FP_m$, we can see that the more rows there are, the slower the FPR reaches its asymptotic (saturated) value.

	%Note that the previous statement is valid for both SQF and QHT. 
A similar phenomenon can be observed with the Stable Bloom Filter \cite{Den06}\footnote{In a Stable Bloom Filter, the FPR is bounded by $\displaystyle\left(1-\left(\frac 1 {1 + \frac{1}{P\left(1/K - 1/m\right)}}\right)^{\text{Max}}\right)^K$, where $m$ is the number of rows, and $K$, $P$, $\text{Max}$ are diverse filter parameters. We clearly see that a higher number of rows will only decrease the FPR down to a certain point, but no further.}: adding rows will only decrease the FPR to a certain point. 
An other structure adapted to streaming data we found, the block decaying Bloom Filter \cite{She08}, operates on a sliding window and therefore did not use the same False Positive definition as the one we did\footnote{More precisely, their definition false positive definition is restrained to the sliding window. So an element is a false positive if is not already present \emph{in the sliding window}, and yet marked as a duplicate.}.

\subsubsection{Application to SQF}\label{sub:fpr_sqf}

One should be tempted to directly apply this result to an SQF. However, as pointed out earlier, in an SQF fingerprints are correlated and therefore not equiprobable\footnote{For example, consider the SQF with the parameters $r=3$, $r' = 1$. Only the hash $h^r_1 = 000$ will lead to the fingerprint $000$, whereas both hashes $h_2^{r} = 001$ and $h_3^{r} = 010$ will lead to the fingerprint $001$.}.

However, for an optimal SQF, fingerprints are equiprobable so the analysis above holds for optimal SQFs. In the general case, the authors of \cite{Dut13} have approximated the probability of fingerprint collision in an SQF with $\frac{1}{2r'\sqrt{\pi r}}$. Replacing $\frac 1S$ with this probability in our analysis, we get an \emph{approximate} asymptotic FPR of $\frac{k}{2^{r'}\sqrt{\pi r}}$ for SQFs, as announced in Section~\ref{sub:sqf_errors}.

	\subsection{False Negative Rate}\label{sub:FNR}
	\begin{theorem}
	For a QHT of $N$ rows, $k$ buckets per row, and $S$ different fingerprints, assuming $U \gg N$, then as the number of insertions goes to infinity, $\displaystyle\FNR_{\infty} = 1 - \frac{k}{S}$.
	\end{theorem}
	
	\begin{proof}
	Assume that $e$ is a duplicate, we denote by $e_k$ the last element in the stream such that $e_i = e$. Following the same reasoning as for the FPR, $e_i$ will trigger a false negative if and only if $e_i$ is removed from the filter before $e$ arrives, and if any false duplicate of $e$, inserted between the removal of $e_i$ and $e$, is deleted before $e$ arrives.
	
	Let us assume that $e_i$ is deleted at time $j$ (this happens with probability $\displaystyle P_{i,j}^\text{Del}$) by something else than a false duplicate. Using similar arguments than in the previous section, we have $P_{i,j}^\text{Del} = \left(1- \frac{S - k}{k(NS-1)}\right)^{j-i-1} \frac{S - k}{k(NS-1)}$. 
    
    The probability that all false duplicates, inserted after time $j$, are deleted before $e$ arrives is $1 - \FP_{n-j}$.
	
	Let $a = \frac{S-k}{k(NS - 1)}$, $b=\frac{k}{S}$ and $c = 1 - \frac{1}{Nk}$, and denote by $\FN_{i,m}$ the probability, in a stream $E = (e_1, \dotsc, e_m, e)$ where the last duplicate of $e$ is inserted at $i$, that $e_i$ is deleted before $e$ arrives and no false positive persists until $e$,
	\begin{align*}
		\FN_{i,m} = {} & \sum_{j=i+1}^m P_{i,j}^\text{Del} (1-\FP_{m-j}) \\
		 = {}& \sum_{j=i+1}^m \left(1- a\right)^{j-i-1} a \left(1 - b\left(1-c^{m-j}\right)\right)\\
		 = {} & a\left(1-b \right)\sum_{j=i+1}^m \left[\left(1- a\right)^{j-i-1}\right] \\ & 
         	+ ab\sum_{j=i+1}^m \left[\left(1- a\right)^{j-i-1}c^{m-j}\right]\\
	\end{align*}
	Given that $1 - a \neq c$, we get:
	\begin{align*}
		\FN_{i,m} = {}& \left(1-b\right) \left(1 - \left(1-a\right)^{m-i}\right) + ab \frac{\left(1-a\right)^{m-i} - c^{m-i}}{1 - a - c}
	\end{align*}
    The probability $\FN_m$ of $e$ to be a false negative is then $\displaystyle\sum_{i=1}^m P_{\text{dup}, i} \cdot \FN_{i,m}$, where $P_{\text{dup}, i}$ is the probability that the last duplicate already seen is $e_i$, so $P_{\text{dup}, i} = \left(\frac{U-1}{U}\right)^{m-i}\frac{1}{U}$.
    
    We obtain the probability $\FN_m$ that the $(m+1)$\textsuperscript{th} element $e$ of the stream will be a false negative:
	\begin{align*}
		\FN_m & = \sum_{i=1}^m P_{\text{dup}, i} \cdot \FN_{i,m}
	\end{align*}
    For a stream of $n$ elements, the false negative rate $\FNR_n$ is defined as the average error probability: $\FNR_n = \frac{1}{n} \sum_{m=1}^{n} \FN_m$
    
    We show, in Appendix~\ref{app:FN_derivation}, that for $\displaystyle\FNR_\infty = \lim_{n \to \infty} \FNR_n$,
	\begin{align*}
		\FNR_\infty = {} & 1-b - \frac{1-b}{U - (U-1)(1-a)} \\
				 & + \frac{ab}{(1-a-c)(U - (U-1)(1-a))} \\
				 & - \frac{ab}{(1-a-c)(U - (U-1)c)}
	\end{align*}
    A not obvious consequence of the above expression for FNR is that when $N$ increases, which corresponds to using more memory, it is possible to achieve an arbitrary small error rate.\footnote{This fact is straightforward but requires the substitution of $a, b, c$ by their full expression, before taking the limit.} At some point however, the memory is so large that every stream element can be stored, and there is no need for probabilistic data structures anymore. In practice we expect memory to be a limited resource, so this situation is unlikely to present itself.
	
	Finally, assuming $U \gg Nk$ (or even $U \gg N$), i.e., that there are more distinct elements in the stream that what the filter is able to store, we get the limit rate of $1 - b$, which is $\FNR_\infty \approx 1 - \frac{k}{S}$.
\end{proof}
	Note that $\FNR_\infty + \FPR_\infty = 1$.

The value of $\FNR_\infty$ also gives (using the same corrections as for the FPR in Section~\ref{sub:fpr_sqf}) the FNR for SQF. % C'est un peu moche ce (1) mais ok / C'est leur style de LaTeX :p

    \subsection{Error Rate and Filter Saturation}\label{sub:saturation}
    
    \begin{claim}
    The asymptotic relation $\FNR_\infty + \FPR_\infty = 1$ is \emph{universal} (as long as $U \gg M$) amongst hash-based duplicate detection filters and is the result of the filter's saturation. Furthermore, when a filter reaches saturation, it behaves similarly, from the error rate point of view, to a filter answering \DUPLICATE at random (we will call such filters random filters).
    \end{claim}
    
    \begin{proof}
    In order to see this, first note that random filters always verify the relation $\FNR + \FPR = 1$. Given that a random filter will return \DUPLICATE with a probability of $p$, an unseen element will be classified as \DUPLICATE with probability $p$, and a duplicate will be classified as \UNSEEN with probability $1-p$, hence the result.
    
    Now, on infinite streams with infinitely many different elements, filters are saturated with information. Given that a filter can only store at most one element per bit (cf. Section~\ref{sub:dup_detection}), thus a filter of size $M$ can remember at most $M$ elements. However, after $fM$ insertions (with $f \gg 1$), the filter remembers at most a tiny fraction $\frac 1f \ll 1$ of the stream. Having reached its saturated state, the filter has an extremely tiny probability $\frac 1f \approx 0$ of correctly guessing whether the incoming element is a duplicate or not, given what the filter actually knows about the stream. For this reason, the best strategy of a saturated filter is almost indistinguishable from a random strategy, i.e. randomly outputting \DUPLICATE. When the stream grows indefinitely, the filter becomes asymptotically equivalent to a random filter.
    \end{proof}
    
     Furthermore, in practical cases, we observe that duplicate filters are equivalent to a specific kind of random filters: they answer \DUPLICATE with some fixed probability $p$, $p$ depending on the filter's nature and its parameters (i.e., $p$ does not change with time).
    
    \paragraph{Interpretation}
    The interpretation of these results could suggest that streaming filters are useless: they need more memory than a random filter, despite being asymptotically equivalent. However, this is only true because of our hypotheses and definitions: we define a false negative to be a duplicate element claimed as unseen by the filter. However, after some amount of time, it is often acceptable that the element may be considered as unseen again: for instance a nonce is theoretically unique, but in practice after a reasonable amount of time nonce reuse is not a vulnerability. For this reason, adapting false positive and false negative to sliding windows may be relevant here. Moreover, we assumed that all elements of the stream had the same probability of occurrence. In practice, this hypothesis is not always correct, and as we will see in Section~\ref{sub:benchmark}, filters operating on real data perform significantly better than random filters, and resist better to saturation.
	    
    \subsection{Comparing QHT to SQF}\label{sub:application_sqf}
    
    Let us compare the memory required for a QHT to reach the same error rates than an SQF. 

    Given that both FPR and FNR depend only on $k$, $N$ and $S$, imposing the equality on these parameters ensures that both filters have exactly the same FPR and FNR.
    
    Note that $S$ is not a user-chosen parameter, but rather a consequence of other parameters.

    \begin{theorem}
    For exactly the same FPR and FNR, a QHT requires 33\% fewer memory than an SQF.
    \end{theorem}
    The proof is made through the following subsections.
    
    \subsubsection{Deriving S from Filters Parameters}
    
    For a QHT, $S$ is derived from the number of bits of the fingerprint $s$, with the straightforward relation\footnote{If we are on a system without the \textsf{empty} feature (see Section~\ref{sub:empty}, then $S = 2^s - 1$. For an SQF, one can just assign the \textsf{empty} value to one of the unassigned fingerprints.}: $S = 2^s$. For an SQF, however, the relation is more complicated.

    When $r$ and $r'$ are fixed, for any element $e$, let $h_e^r$ be decomposed as $h_e^r = h_e^{r'} \| f$, where $h_e^{r'}$ is the $r'$-bits word used in the fingerprint, $h$ being the $r - r'$ remaining bits of $h_e^r$. For $\omega(\cdot)$ the Hamming weight function, we have $s(e) = h_{r'} \| \omega(h_e^r)$. Yet $\omega(h_e^r) = \omega(h_e^{r'}) + \omega(f)$. We know that $\omega(h_e^{r'})$ is entirely dependent on $h_e^{r'}$, which is already used in the fingerprint. Thus, if we fix $h_e^{r'}$, there are only $r - r' + 1$ possible values for $\omega(f)$ and thus for $\omega(h_e^r)$. Given that $h_e^{r'}$ can have $2^{r'}$ different values, we get that $S_{SQF} = 2^{r'} \cdot (r-r'+1)$.
    
    \subsubsection{Comparing Required Memory}
   For QHTs, we have the relation $M_\text{QHT} = N_\text{QHT}k_\text{QHT}s$. For SQF, the formula is rather $M_\text{SQF} = N_\text{SQF}k_\text{SQF}(r' + \lceil \log_2(r+1) \rceil)$ (because fingerprints occupy $r' + \lceil \log_2(r+1) \rceil$ bits).
   
  Given that $N_\text{QHT} = N_\text{SQF}$ and $k_\text{QHT} = k_\text{SQF}$, the ratio $\frac{M_\text{QHT}}{M_\text{SQF}}$ is:
   \begin{align*}
   \frac{M_\text{QHT}}{M_\text{SQF}} &=  \frac{N_\text{QHT}k_\text{QHT}s}{N_\text{SQF}k_\text{SQF}(r' + \lceil \log_2(r+1) \rceil)} = \frac{s}{r'+\lceil \log_2(r+1) \rceil}
   \end{align*}
   Given that $s = \lceil \log_2(S_\text{QHT}) \rceil = \lceil \log_2(S_\text{SQF}) \rceil = \lceil \log_2(2^{r'}\cdot(r-r'+1)) \rceil$, we have
   \begin{align*}
   \frac{M_\text{QHT}}{M_\text{SQF}} &= \frac{\lceil \log_2(2^{r'}(r-r'+1)) \rceil}{r'+\lceil \log_2(r+1) \rceil}
   									 = \frac{r' + \lceil \log_2(r-r'+1)\rceil}{r'+\lceil \log_2(r+1) \rceil}
   \end{align*}
   Using the recommended settings in \cite{Dut13}\footnote{Their optimal choice of parameters also includes setting $k=4$. However, setting $r=2$ and $r'=1$ imposes $S=4$, so that setting $k=4$ results in an FPR of 1: after some time the filter systematically responds \DUPLICATE. Using e.g. $k=3$ avoids this.} ($r = 2$, $r' = 1$) the ratio becomes $\frac{M_\text{QHT}}{M_\text{SQF}} = \frac{2}{3}$
   %\begin{align*}
   %\frac{M_\text{QHT}}{M_\text{SQF}} &= %\frac{2}{3}
   %\end{align*}
   , which concludes the proof. 

    \subsection{Parameter Tuning}\label{sub:tuning}
    
    As noted in Section~\ref{sub:FNR}, no matter the choice of the parameters we have $\FNR + \FPR = 1$: any particular parameters choice will be a trade off between good FPR and good FNR performance, at least asymptotically. However, when the stream is small enough, one may choose parameters that will maximally delay saturation.
    
    We know that $M = Nk\log_2(S)$, so plugging this into the FPR formula gives
  \begin{align*} \FPR_m & = \frac{k}{S}\left(1-\left[1-\frac{1}{kN}\right]^m\right) \\ & = \frac{k}{S}\left(1-\left[1-\frac{\log_2(S)}{M}\right]^m\right)
  \end{align*}
    
    Which means that, for fixed $M$ and $\frac{k}{S}$ (i.e. for a fixed memory amount and a fixed asymptotic FPR), $\log_2(S)$ must be as small as possible in order to keep the FPR low for as long as possible.
    
    For instance, assume that we want an asymptotic FPR of 25\%. The potential values for the couple $(k, S)$ are $(1, 4)$, $(2, 8)$, $(3, 16)$ and so on (leading respectively to $s = 2, 3, 4$). Because of the above relation, we know that setting $k=1, s=2$ will yield the best saturation resistance for the FPR.

   In order to test this heuristic, we compared several QHT of approximately 65,536 bits, with the same ratio $\frac kS$, but a different value for $S$. We took streams of 100,000 elements (well under saturation value for this amount of memory), from an alphabet of $2^{20}$ elements. Each element of the stream was uniformly randomly selected, leading to a stream with about 4.6\% of duplicate elements. We averaged the results on 10 runs.
   
   We observe in Table~\ref{tab:comparison-S} that filters with a small value of $S$ do indeed perform better than filters with a bigger value of $S$, which concludes the experiment.

\begin{table}[H]
	\small
	\centering
    \caption{Error rates of QHTs with the same asymptotic FPR}
	\label{tab:comparison-S}
	\begin{tabular}{cccccc}
		 & $S = 4$ & $S = 8$ & $S = 16$ & $S = 32$ & $S = 64$\\ \hline
		$\FPR$ (\%) & 22.57 & 23.25 & 23.53 & 23.62 & 23.50 \\
		$\FNR$ (\%) & 35.89 & 44.24 & 50.77 & 54.55 & 58.73 \\ \hline
		$\FPR + \FNR$ & 58.45 & 67.49 & 74.30 & 78.17 & 82.23 \\

	\end{tabular}
\end{table}
    
    \section{Further improvements: QQHTD}
    
    \subsection{Keeping Track of Duplicates}\label{sub:qhtd}
    In Algorithm~\ref{a:QHT}, we do not insert anything if the element is detected as a duplicate. However, following \cite{Fan14}'s example, we can insert it anyway, resulting in a structure we call QHT with Duplicates, or QHTD. We briefly discuss its properties.

    As we showed previously, the asymptotic FPR of a QHT is $\frac{k}{S}$, which was expected: each cell stores $k$ distinct fingerprints, the probability that one of them matches the fingerprint of a unique element is logically $\frac{k}{S}$. Similarly, in QHTD each cell stores $k$ fingerprints, not necessarily distinct. The probability that at least one of these fingerprints is the same than the one of an unseen element is $\text{FPR}_\text{QHTD} = 1 - \left(1-\frac{1}{S}\right)^k$. Given the results of Section~\ref{sub:saturation}, the asymptotic FNR of QHTD is $\left(1-\frac{1}{S}\right)^k$.
    
    \subsection{Queuing Buckets for a Better Sliding Window}
    One caveat of the QHT (and QHTD) is the fact that at any insertion, any element of the row is equally likely to be evicted: if this allows an easy FPR and FNR derivation, it makes it functioning a bit counter intuitive. As a matter of fact, one would expect a filter to first forget about the oldest elements before forgetting about the newest ones. Indeed, this behaviour matches the need of a filter operating on a \emph{sliding window}, without taking into account oldest elements.
    
    The solution we provide for QHT is to order the buckets of a given cell in a FIFO queue, which means that instead of selecting a random bucket in the cell for insertion, one will append the fingerprint to the end of the queue, and pop the first element (so that the size of the queue remains constant). Combined with QHTD improvement, this yields Algorithm~\ref{alg:QHT3}.
    
    \begin{algorithm}[H]
	\caption{Queued Quotient Hash Table with Duplicates' (QQHTD) \textsf{Stream}}\label{alg:QHT3}
	\begin{algorithmic}[1]
		\For{each element $e \in E$}
        	\State result $\gets$ \UNSEEN
			\State Quotient of $e$: $h(e)$; Fingerprint of $e$: $s(e)$.
			\For{each bucket $b_i$ in the queue $T[h(e)]$}
				\If{(entry in $b_i) = s(e)$}
					\State result $\gets$ \DUPLICATE
                    \State \textbf{break}
				\EndIf
			\EndFor
			\State Pop the first element of the queue $T[h(e)]$ and append $s(e)$ at the end of same queue
			\State \Return result
		\EndFor
	\end{algorithmic}
\end{algorithm}
    
    Note that classical queues (i.e. doubly chained lists) are not suited for our use, as chains require extra storage bits for pointers. Thus, we create an array of $k$ elements, in which we manually move every element at each ``pop''. When $k$ is small (typically less than 5, which it usually is), the added overhead is not significant.
    
    Finally, note that a QQHTD with one bucket par cell is equivalent to its QHT counterpart. However, with more buckets per cell, QQHTDs offer a noticeable improvement over QHTs on real data streams (see Appendix~\ref{app:tables}). 

	\section{Benchmarks}\label{sub:benchmark}
	
		\begin{table*}[ht]
		\caption{Error rate (multiplied by $100$) on streams of 150,000,000 elements}
		\label{tab:errors}
		
		\begin{tabularx}{\textwidth}{cc *{7}{Y}}
			\hline
			
			Stream (duplicate \%)					& Memory (bits)	& SQF	& QHT QQHTD	& Cuckoo	& SBF	& A2	& b\_DBF	\\ \hline
			\multirow{4}{*}{Real (10.3 \%)}			& 8e+06		& 51.25	& 43.78	& 66.48	& 54.08	& 58.94	& 45.22	\\
			& 1e+06		& 55.18	& 48.38	& 68.96	& 57.17	& 62.12	& 45.30	\\
			& 100,000	& 58.21	& 52.22	& 71.83	& 59.44	& 64.88	& 59.39	\\
			& 10,000		& 66.45	& 58.75	& 78.86	& 76.29	& 66.42	& 99.77	\\ \hline
			\multirow{4}{*}{Artificial (88.82 \%)}	& 8e+06		& 86.49	& 82.76	& 96.94	& 97.79	& 88.06	& 99.96	\\
			& 1e+06		& 98.27	& 97.80	& 99.60	& 99.74	& 98.49	& 99.96	\\
			& 100,000	& 99.78	& 99.79	& 99.96	& 99.98	& 99.84	& 99.96	\\
			& 10,000		& 99.97	& 99.97	& 100.02& 99.99	& 100.00	& 99.98	\\ \hline
			\multirow{4}{*}{Artificial (39.79 \%)}	& 8e+06		& 96.51	& 95.37	& 99.21	& 99.40	& 96.86	& 99.99	\\
			& 1e+06		& 99.56	& 99.42	& 99.91	& 99.91	& 99.60	& 99.99	\\
			& 100,000	& 99.94	& 99.95	& 100.00	& 99.99	& 99.96	& 99.98	\\
			& 10,000		& 100.00	& 100.00	& 100.00	& 100.00	& 99.99	& 100.00	\\ \hline
		\end{tabularx}
	\end{table*}
	
	\subsection{Comparison of QHT to Other Filters}
    We used streams of 150,000,000 elements, on filters of size ranging from 10 kb to 8 Mb. In any case the filters are too small to keep track of the whole stream, and we will see that filters do reach saturation. We used 2 artificial streams, for which the elements where randomly generated from an alphabet of $2^{24}$ and $2^{27}$ elements respectively, leading to a duplicate rate of about $88\%$ and $38\%$ respectively. We also used a real dataset of URLs visited by a crawling robot, extracted from the April 2018 CommonCrawl's dump~\cite{CommonCrawl} The source code is available from on BitBucket.\footnote{\url{https://bitbucket.org/team_qht/qht/src/master}}.
    
    The filters, so their asymptotic FPR was as close as possible to the arbitrary value of 25\%, are: 
\begin{itemize}
	\item SQF, 1 bucket per row, $r = 2$ and $r' = 1$
	\item QHT, 1 bucket per row, 3 bits per fingerprint. This specific QHT is equivalent to a QQHTD with the same parameters, so we do not include the latter in the benchmark. 
	\item Cuckoo Filter\cite{Fan14}, cells containing 1 element of 3 bits each
	\item Stable Bloom Filter (SBF) \cite{Den06}, 2 bits per cell, 2 hash functions, targeted FPR of 0.02\footnote{Our benchmarks actually obtained an asymptotic FPR of around 28\%, without us being able to find bugs in our implementation.}
    \item $A2$ Filter\cite{Yoo10}, targeted FPR of $0.1$ on the sliding window.
	\item Block-Decaying Bloom Filter (b\_DBF)\cite{She08}, sliding window of 6000 elements.

\end{itemize}

Results, averaged on 5 runs, are given in Table~\ref{tab:errors}. Note that, for better readability, the error rates have been multiplied by $100$ in the table. Further more, we recall that the error rate, being defined as $E = \FPR + \FNR$, is bounded by $0$ below and $2$ above, $1$ being the error rate of a random filter. A filter can have worse results than random; for instance a filter which is always wrong has an error rate of $2$.

     As we can see in Table~\ref{tab:errors}, QHT (or QQHTD) are extremely competitive and resist very well to saturation; they also appear to be the most competitive on the real stream. b\_DBF are efficient, but reach very quickly their saturation. A more detailed analysis of their FPR (see Appendix~\ref{app:tables}) shows that even though their FPR is close to 0, they get an FNR close to 1. Moreover, as we see in Section~\ref{sub:speed}, they are significantly slower, which can be a bottleneck for critical applications.
     Further more, not only are QHT/QQHTD the most efficient filter on both real and artificial streams, they are also very easily tunable, and any asymptotic FPR rate is very simply achievable. This is not the case of other filters, such as A2 or b\_DBF, which require careful tuning. 

	\subsection{Speed Comparison}\label{sub:speed}
    We also benchmarked the speed of every filter on real-time detection, on a laptop with Intel i7. We used the same filters as in the previous subsection, with a memory of 1 Mb, on a stream of 150,000,000 elements. We averaged, on these filters, the time needed for $\textsf{Insert} \circ \textsf{Detect}$ to execute for each element. Results are shown in Table~\ref{tab:time}. We observe that safe for SQF, QHT is 6 times as fast as any other filter, and 10 times as fast as b\_DBF. Even SQF is 50\% slower than QHT. Even with additional features, QQHTD are also faster than SQF by a large magin, because fingerprint derivation is more costly in the latter. As a conclusion, we observe that QQHTD are most suited filters for high-speed analysis.

\begin{table*}[]
	\centering
	\caption{Average amount of time (in $µ$s) required for one iteration of $\mathsf{Stream}$ on each filter with 1 Mb of memory} \label{tab:time}
	
	\begin{tabular}{c|ccccccc}
			%\hline
			Filter    &SQF        & QHT   & QQHTD  & Cuckoo    & SBF     & A2      & b\_DBF \\
			\hline
			Time &0.423 & 0.288 & 0.330 & 2.464 & 1.578 & 1.280 & 2.565 \\
			%\hline
	\end{tabular}
\end{table*}

    \section{Adversarial Resistance}
    Now, despite the good performances of QQHTD on normal streams, one may not always assume that the stream is \enquote{normal}: there may be an attacker trying to fool the filter. For instance if the filter must detect duplicates in order to avoid an attack (nonce requirements), then the question of adversarial resistance is primordial.
    
    We can model an adversarial system in which an attacker has a knowledge of the output of \textsf{Stream} (i.e., whether the element is classified as \DUPLICATE or \UNSEEN), but no knowledge of the internal memory state $\mathcal{M}$. Thus the attacker is able to carry an adaptive attack, by choosing the next element to send to the filter as a function of all previous insertions. In this adversarial game, the attacker can send an arbitrary stream to the filter, and is allowed to get the result of \textsf{Stream} for every element. Then, at her convenience, the attacker goes into the second phase of the game, in which she has two possible actions:
    
    \begin{itemize}
    \item Send an unseen element that will be a false positive with high probability (false positive attack);
    \item Send a duplicate element that will be a false negative with high probability (false negative attack).
    \end{itemize}
    
    \begin{theorem}
    No filter can resist a false negative attack.
    \end{theorem}
    
    \begin{proof}
    We craft false negatives in $\Omega(M)$ steps ($M$ being the filter's memory size). Let us remind that no structure can remember more than one element per memory bit. For this reason, the structure can remember at most $M$ different elements. Consequently, if the attacker generates a stream of random unseen elements, then on average each element will stay for $M$ insertions in the filter's memory. More generally, after $hM$ insertions (for some rational $h \geq 1$), an element is forgotten with probability at least $P_\text{forgot} = 1 - (1-\frac{1}{M})^{hM} \simeq 1 - e^{-h}$.
    Thus an attacker simply generates $\Omega(M)$ unique elements before sending the first element again. $M$ can even be estimated via saturation (see Section~\ref{sub:saturation}). Given that CPU time is cheaper than memory requirements, the attacker keeps her advantage over any filter of any size.
    \end{proof}

    \begin{theorem}
    Assuming the existence of one-way functions, QHT can resist a false positive attack.
    \end{theorem}
    
    \begin{proof}
    Following \cite{10.1007/978-3-662-48000-7_28} we replace all hash functions by one-way hash functions, and apply a (secret) one-way permutation on incoming elements, then classically store the results in the filter. Because of the permutation, the attacker gains no advantage in adaptively choosing the elements, thus loosing her advantage. 
    \end{proof}
    
    As a conclusion, QHTs are adapted to contexts where low false positives are crucial, such as white-list email filtering. Note however that this is the case of most filters, as long as they rely on hash functions (and so can apply \cite{10.1007/978-3-662-48000-7_28}).
	
	\section{Conclusion}
	This paper introduces a new duplicate detection filter, QHT, and its variant QHTD. QHTs achieve a better utilization of the available space, and as such are more efficient than existing filters. Moreover, QHTD have more efficiency for detecting duplicates in a real dataset.
	
	We showed that, for an infinite stream with an infinite number of unseen elements, the number of rows is less important than the fingerprint space, and the number of buckets per row. Moreover, we proved that all filters, having reached saturation, are not more efficient than random filters, and as such, a benchmarking of stream filters should only focus on the pre-saturation state, with small streams.
	
    Even though QHTs are significantly more efficient than other structures in the literature, we do not know if these filters are optimal: are there other filters having an optimal resilience to saturation? Future work also includes examining the theoretical resistance to saturation of QQHTDs, and a finer examination of the QHT/QQHTD behaviour on a sliding window.
   	
\appendix

\section{Tables of error rates for various streams}\label{app:tables}

 Table~\ref{tab:app_long} gives the error rates (FPR and FNR) used to derive Table~\ref{tab:errors}.

\begin{table*}[!t]
 \caption{Results of filters for various streams of 150,000,000 elements (FPR/FNR in \%)}
 \label{tab:app_long}

     	\begin{tabularx}{\textwidth}{cc *{7}{Y}}
         \hline

 		Stream (duplicate \%)					& Memory (bits)	& SQF		& QHT/QQHT	& Cuckoo		& SBF			& A2			& b\_DBF	\\ \hline
 		\multirow{4}{*}{Real (10.3 \%)}			& 8e+06		& 24.61/26.64	& 14.00/29.78	& 28.23/38.26	& 25.10/28.98	& 37.72/21.21	& 0.00/45.22	\\
 												& 1e+06		& 24.95/30.22	& 14.25/34.14	& 28.30/40.65	& 25.23/31.94	& 38.00/24.11	& 0.15/45.16	\\
 												& 100,000	& 24.99/33.22	& 14.28/37.94	& 28.33/43.51	& 25.42/34.03	& 38.05/26.83	& 25.75/33.64	\\
 												& 10,000		& 25.00/41.45	& 14.29/44.46	& 28.37/50.49	& 26.00/50.29	& 38.01/28.41	& 99.58/0.20	\\ \hline
 		\multirow{4}{*}{Artificial (88.82 \%)}	& 8e+06		& 21.87/64.62	& 12.02/70.74	& 27.80/69.14	& 25.94/71.85	& 35.22/52.84	& 0.00/99.96	\\
 												& 1e+06		& 24.60/73.67	& 14.00/83.80   & 28.41/71.19	& 26.43/73.30	& 37.72/60.77	& 0.17/99.79	\\
 												& 100,000	& 24.94/74.84	& 14.26/85.53	& 28.52/71.44	& 26.49/73.49	& 38.00/61.83	& 27.46/72.51	\\
 												& 10,000		& 24.98/74.99	& 14.28/85.69	& 28.55/71.47	& 26.50/73.50	& 38.03/61.97	& 99.67/0.31	\\ \hline
 		\multirow{4}{*}{Artificial (39.79 \%)}	& 8e+06		& 24.42/72.09	& 13.86/81.52	& 28.39/70.82	& 26.40/73.00   & 37.54/59.32	& 0.00/99.99	\\
 												& 1e+06		& 24.92/74.64	& 14.24/85.18	& 28.50/71.41	& 26.47/73.44	& 38.00/61.60	& 0.17/99.82	\\
 												& 100,000	& 24.99/74.96	& 14.29/85.66	& 28.51/71.49	& 26.49/73.50	& 38.05/61.91	& 27.46/72.52	\\
 												& 10,000		& 25.00/74.99   & 14.28/85.72	& 28.52/71.49	& 26.49/73.51	& 38.02/61.97	& 99.69/0.31	\\ \hline
         \end{tabularx}
     \end{table*}
    
    \section{Deriving $\FNR_\infty$}\label{app:FN_derivation}
    This derivation was removed from Section~\ref{sub:FNR} for better readability. We know that $\displaystyle\FNR_n = \frac{1}{n}\sum_{m=1}^{n}\FN_m = \frac 1n \sum_{m=1}^n \sum_{i=1}^m P_{\text{dup}, i} \cdot \FN_{i,m}$, so
    	
    \begin{align*}
    \FNR_n & = \frac{1}{n}\sum_{m=1}^{n}\sum_{i=1}^{m} \left(\frac{U-1}{U}\right)^{m-i}\frac{1}{U} \left[(1-b) \vphantom{\frac{c^m}{c}}
        	\right. \\& \qquad \left.
        	\left(1-(1-a)^{m-i}\right)
        	%\right. \\& \qquad\qquad\qquad \left.
        	+ ab \frac{(1-a)^{m-i} - c^{m-i}}{1-a-c} \right]
    \end{align*}
	
	Expanding,
    
    \begin{align*}
    \FNR_n &= \frac{1}{nU} \sum_{m=1}^n \left[(1-b) \sum_{i=1}^m\left(\frac{U-1}{U}\right)^{m-i}
        	\right. \\& \qquad \left.
        	+ (1-b) \sum_{i=1}^m \left(\frac{U-1}{U} \cdot (1-a)\right)^{m-i}
			\right. \\& \qquad \left.
			+ \frac{ab}{1-a-c} \sum_{i=1}^m \left(\frac{U-1}{U}(1-a)\right)^{m-i}
			\right. \\& \qquad \left.
			+ \frac{ab}{1-a-c} \sum_{i=1}^m \left(\frac{U-1}{U}c\right)^{m-i}
        	\right]\\
    \end{align*}\begin{align*}
	\FNR_n ={}& \frac{1}{nU} \sum_{m=1}^n \left[
        (1-b)\frac{1-(1-1/U)^m}{1/U} 
        	\right. \\& \qquad \left.
			+ (1-b)\frac{1-\left((1-1/U)(1-a)\right)^m}{1-(1-a)(U-1)/U}
            \right. \\& \qquad \left.
            + \frac{ab}{1-a-c}\frac{1-\left((1-1/U)(1-a)\right)^m}{1-(1-a)(U-1)/U}
            \right. \\& \qquad \left.
            - \frac{ab}{1-a-c}\frac{1-\left((1-1/U)c\right)^m}{1-c(U-1)/U}
        \right]
    \end{align*}

    Given that $\displaystyle \FNR_\infty = \lim_{n\to \infty} \FNR_n$, using Ces\`{a}ro's mean we get:
    
    \begin{align*}
    \FNR_\infty ={}& 1-b - \frac{1-b}{U - (U-1)(1-a)}
     	\\ & 
        + \frac{ab}{(1-a-c)(U - (U-1)(1-a))}
        \\ &
        - \frac{ab}{(1-a-c)(U - (U-1)c)}
    \end{align*}

Which concludes the proof.

\section{Comparison of QHT and QQHTD}\label{app:comparison}

In this appendix, we explore the difference between a QHT and a QQHTD with the same parameters, on the same stream. We took filters of 65,536 bits each, on streams of 100,000 elements each. One stream issued from our `real' dataset (10.32\% of duplicates), the other a random uniform stream on an alphabet of $2^{20}$ elements (4.62\% of duplicates). Results are given in \ref{tab:comparision}.

We observe that while QQHTD offer no advantage on artificial streams, their performance (relative to the QHT) are noticeably better, which empirically validates the optimizations.

\begin{table*}[]
	\centering
	\caption{Error rate (times $100$) of QHTs and QQHTDs with the same parameters on different streams, depending on their parameters $k$ and $S$}
	\label{tab:comparision}
	
	\begin{tabular}{cccccc}
		%\hline
		 Stream (duplicate \%) & Filter & $k = 2, S = 8$ & $k = 4, S = 16$ & $k = 8, S = 32$ & $k = 16, S = 64$ \\ \hline
		\multirow{2}{*}{Real (10.3 \%)} & QHT & 27.01 & 28.07 & 28.95 & 29.3 \\
										& QQHTD & 24.67 & 24.56 & 24.42 & 24.62 \\ \hline
	    \multirow{2}{*}{Artificial (4.62 \%)} & QHT & 67.39 & 74.48 & 79.32 & 82.10 \\
											  & QQHTD & 67.91 & 74.41 & 79.19 & 82.26 \\
	\end{tabular}
\end{table*}

\bibliographystyle{acm}
\bibliography{qht_bib} 

\begin{thebibliography}{10}

\bibitem{Babcock:2004:LSA:977401.978165}
{\sc Babcock, B., Datar, M., and Motwani, R.}
\newblock Load shedding for aggregation queries over data streams.
\newblock In {\em Proceedings. 20th International Conference on Data
  Engineering\/} (March 2004), {IEEE} Computer Society, pp.~350--361.

\bibitem{Blo70}
{\sc Bloom, B.~H.}
\newblock Space/time trade-offs in hash coding with allowable errors.
\newblock {\em Commun. ACM 13}, 7 (July 1970), 422--426.

\bibitem{Borg:2014:RSD:2652524.2652556}
{\sc Borg, M., Runeson, P., Johansson, J., and M\"{a}ntyl\"{a}, M.~V.}
\newblock A replicated study on duplicate detection: Using apache lucene to
  search among android defects.
\newblock In {\em Proceedings of the 8th ACM/IEEE International Symposium on
  Empirical Software Engineering and Measurement\/} (New York, NY, USA, 2014),
  ESEM '14, ACM, pp.~8:1--8:4.

\bibitem{Che17}
{\sc Chen, H., Liao, L., Jin, H., and Wu, J.}
\newblock The dynamic cuckoo filter.
\newblock In {\em 2017 IEEE 25th International Conference on Network Protocols
  (ICNP)\/} (Oct 2017), pp.~1--10.

\bibitem{Coh03}
{\sc Cohen, S., and Matias, Y.}
\newblock Spectral bloom filters.
\newblock In {\em Proceedings of the 2003 ACM SIGMOD International Conference
  on Management of Data\/} (New York, NY, USA, 2003), SIGMOD '03, ACM,
  pp.~241--252.

\bibitem{Cor09}
{\sc Cormode, G., and Muthukrishnan, S.}
\newblock An improved data stream summary: the count-min sketch and its
  applications.
\newblock {\em Journal of Algorithms 55}, 1 (2005), 58 -- 75.

\bibitem{Den06}
{\sc Deng, F., and Rafiei, D.}
\newblock Approximately detecting duplicates for streaming data using stable
  {Bloom} filters.
\newblock In {\em {SIGMOD} Conference\/} (2006), {ACM}, pp.~25--36.

\bibitem{Dut13}
{\sc Dutta, S., Narang, A., and Bera, S.~K.}
\newblock Streaming quotient filter: A near optimal approximate duplicate
  detection approach for data streams.
\newblock {\em Proc. VLDB Endow. 6}, 8 (June 2013), 589--600.

\bibitem{Fan14}
{\sc Fan, B., Andersen, D.~G., Kaminsky, M., and Mitzenmacher, M.~D.}
\newblock Cuckoo filter: Practically better than bloom.
\newblock In {\em Proceedings of the 10th ACM International on Conference on
  Emerging Networking Experiments and Technologies\/} (New York, NY, USA,
  2014), CoNEXT '14, ACM, pp.~75--88.

\bibitem{Fan00}
{\sc Fan, L., Cao, P., Almeida, J., and Broder, A.~Z.}
\newblock Summary cache: A scalable wide-area web cache sharing protocol.
\newblock {\em IEEE/ACM Trans. Netw. 8}, 3 (June 2000), 281--293.

\bibitem{Fu15}
{\sc Fu, M., Feng, D., Hua, Y., He, X., Chen, Z., Xia, W., Zhang, Y., and Tan,
  Y.}
\newblock Design tradeoffs for data deduplication performance in backup
  workloads.
\newblock In {\em Proceedings of the 13th USENIX Conference on File and Storage
  Technologies\/} (Berkeley, CA, USA, 2015), FAST'15, USENIX Association,
  pp.~331--344.

\bibitem{Guo10}
{\sc Guo, D., Wu, J., Chen, H., Yuan, Y., and Luo, X.}
\newblock The dynamic bloom filters.
\newblock {\em IEEE Transactions on Knowledge and Data Engineering 22}, 1 (Jan
  2010), 120--133.

\bibitem{one_dup}
{\sc Jowhari, H., Saglam, M., and Tardos, G.}
\newblock Tight bounds for lp samplers, finding duplicates in streams, and
  related problems.
\newblock {\em CoRR abs/1012.4889\/} (2010), 49--58.

\bibitem{one_dup_opt}
{\sc Kapralov, M., Nelson, J., Pachocki, J., Wang, Z., Woodruff, D.~P., and
  Yahyazadeh, M.}
\newblock Optimal lower bounds for universal relation, and for samplers and
  finding duplicates in streams.
\newblock In {\em {FOCS}\/} (2017), {IEEE} Computer Society, pp.~475--486.

\bibitem{nonce}
{\sc Køien, G.}
\newblock A brief survey of nonces and nonce usage.
\newblock In {\em Securware 2015 - The Ninth International Conference on
  Emerging Security Information, Systems and Technologies"\/} (2015), SECURWARE
  '15, IARIA XPS Press, pp.~85--91.

\bibitem{Metwally:2005:DDC:1060745.1060753}
{\sc Metwally, A., Agrawal, D., and El~Abbadi, A.}
\newblock Duplicate detection in click streams.
\newblock In {\em Proceedings of the 14th International Conference on World
  Wide Web\/} (New York, NY, USA, 2005), WWW '05, ACM, pp.~12--21.

\bibitem{CommonCrawl}
{\sc Nagel, S.}
\newblock April 2018 crawl archive now available, 2018.
\newblock
  \url{http://commoncrawl.org/2018/05/april-2018-crawl-archive-now-available/}.

\bibitem{10.1007/978-3-662-48000-7_28}
{\sc Naor, M., and Yogev, E.}
\newblock Bloom filters in adversarial environments.
\newblock In {\em Advances in Cryptology -- CRYPTO 2015\/} (Berlin, Heidelberg,
  2015), R.~Gennaro and M.~Robshaw, Eds., Springer Berlin Heidelberg,
  pp.~565--584.

\bibitem{She08}
{\sc Shen, H., and Zhang, Y.}
\newblock Improved approximate detection of duplicates for data streams over
  sliding windows.
\newblock {\em Journal of Computer Science and Technology 23}, 6 (2008),
  973--987.

\bibitem{Tar12}
{\sc Tarkoma, S., Rothenberg, C.~E., and Lagerspetz, E.}
\newblock Theory and practice of bloom filters for distributed systems.
\newblock {\em IEEE Communications Surveys Tutorials 14}, 1 (First 2012),
  131--155.

\bibitem{10.1007/3-540-56922-7_6}
{\sc Wolper, P., and Leroy, D.}
\newblock Reliable hashing without collision detection.
\newblock In {\em Computer Aided Verification\/} (Berlin, Heidelberg, 1993),
  C.~Courcoubetis, Ed., Springer Berlin Heidelberg, pp.~59--70.

\bibitem{Yoo10}
{\sc Yoon, M.}
\newblock Aging bloom filter with two active buffers for dynamic sets.
\newblock {\em IEEE Trans. on Knowl. and Data Eng. 22}, 1 (Jan. 2010),
  134--138.

\end{thebibliography}
\end{multicols}

\end{document}